\documentclass[table]{ccjnl}
\usepackage{lipsum,amsmath}
\usepackage{cuted}
\usepackage{bm}
\graphicspath{{figures/}}
\usepackage{multirow}
\usepackage{appendix}
\newtheorem{Prob}{Problem}
\usepackage{tabularx}

\title{High-Performance Transmission Mechanism Design of Multi-Stream Carrier Aggregation for 5G Non-Standalone Network}

\receiveddate{xx. xx, xx}
\reviseddate{xx. xx, xx}
\Editor{xx xx}

\address[1]{School of Communication \& Information Engineering, Shanghai University, Shanghai, 200444, China}
\address[2]{Shanghai Institute for Advanced Communication and Data Science, Shanghai University, Shanghai, 200444, China}
\address[*]{The corresponding author: Shunqing Zhang}

\author{Jun Yu\inst{1, 2}, Shunqing Zhang\inst{1, 2}, Jiayun Sun\inst{1, 2}, Shugong Xu\inst{1, 2}, Shan Cao\inst{1, 2}\\
Email:\{junyu, shunqing, jiayunsun, shugong, cshan\}@shu.edu.cn}

\begin{document}

\maketitle

\begin{abstract}
Multi-stream carrier aggregation is a key technology to expand bandwidth and improve the throughput of the fifth-generation wireless communication systems. However, due to the diversified propagation properties of different frequency bands, the traffic migration task is much more challenging, especially in hybrid sub-6 GHz and millimeter wave bands scenario. Existing schemes either neglected to consider the transmission rate difference between multi-stream carrier, or only consider simple low mobility scenario. In this paper, we propose a low-complexity traffic splitting algorithm based on fuzzy proportional integral derivative control mechanism. The proposed algorithm only relies on the local radio link control buffer information of sub-6 GHz and mmWave bands, while frequent feedback from user equipment (UE) side is minimized. As shown in the numerical examples, the proposed traffic splitting mechanism can achieve more than 90\% link resource utilization ratio for different UE transmission requirements with different mobilities, which corresponds to 10\% improvement if compared with conventional baselines.
\end{abstract}

\section{Introduction} \label{sect:intro}
With the commercial deployment of the fifth-generation wireless communication systems (5G), emerging applications including augmented reality/virtual reality (AR/VR) and high density video streaming have triggered the explosive growth of wireless traffic \cite{1}. Although massive multiple-input multiple-output (MIMO), orthogonal frequency-division multiplexing (OFDM), and advanced link adaptation techniques are quite effective for improving the achievable spectral efficiency \cite{2}, the most powerful tool to improve the overall throughput for each base station (BS) is via carrier aggregation (CA) \cite{3}, i.e., by aggregating multiple available transmission bands together to achieve extremely high data rates. 

In general, the CA technology can be divided into three classes, namely {\em intra-band contiguous CA}, {\em intra-band non-contiguous CA}, and {\em inter-band non-contiguous CA} \cite{DCA}, with one primary component carrier (PCC) and several contiguous or non-contiguous secondary component carriers (SCCs). For intra-band contiguous contiguous or non-contiguous CA, PCC and SCCs share the similar coverage, and the traffic migration among different component carriers is straight forward \cite{5}, while for inter-band non-contiguous CA, the traffic migration task is much more challenging due to the diversified propagation properties of different frequency bands. In the conventional fourth generation wireless systems (4G), the operating frequency bands are below 6 GHz, and the number of supported component carriers are limited to 5, which corresponds to a total bandwidth of 100 MHz. To support enhanced mobile broadband (eMBB) and ultra-reliable low-latency communications (URLLC) of 5G networks, the millimeter wave (mmWave) technology \cite{6,7} has been proposed to deal with the extremely crowded frequency bands below 6 GHz. With more frequency bands available, the number of supported component carriers for 5G CA is increased to 16 with 1 GHz total bandwidth \cite{(6)}.

Since mmWave bands suffer from the high isotropic propagation losses, the corresponding link adaptation schemes, including modulation and coding scheme (MCS) and re-transmission processes, as well as the channel outage events are quite different from sub-6 GHz bands \cite{8}. This leads to several novel designs in mmWave bands, especially when interacted with the above 5G CA technology. For example, a digital pre-distortion technique and a filter bank multicarrier (FBMC) technique have been proposed in \cite{9} and \cite{FBMC,IoT}, respectively, to deal with the linearization issues for extremely wide band power amplifiers. In the physical layer, a novel beamwidth selection and optimization scheme has been proposed in \cite{D2D} to deal with the potential interference among mmWave links and sub-6 GHz links, and the extensions to non-orthogonal multiple access (NOMA) based and multiple input multiple output (MIMO) based transmission strategies have been discussed in \cite{2021carrier} and \cite{CAReceive} as well. In the media access control (MAC) layer, a novel  mechanism to dynamically select sub-6 GHz and mmWave bands has then been proposed in \cite{10}, and later extended to incorporate resource block (RB) level allocation in \cite{CALet}, which guarantees quality-of-experience (QoE) performance for CA enabled users. In the higher layer, collaborative transmission for video streaming applications has been proposed in \cite{deng2020dash}, which allows sub-6 GHz bands for control information and mmWave bands for data transmission.

A more efficient approach to utilize the inter-band non-contiguous CA is to boost transmission rates of mmWave bands with the assistance of sub-6 GHz bands \cite{NCA}. Typical examples include content distribution and processing in vehicular networks \cite{11} and exploiting channel state information (CSI) of a sub-6 GHz channel to choose a mmWave beam \cite{12}. With the recent progresses in achieving high data rates in traffic hotspots, simultaneous transmission via sub-6 GHz and mmWave bands has been proposed in \cite{13} and \cite{14}, which generally requires smart traffic splitting mechanisms in the packet data convergence protocol (PDCP) layer. The above schemes usually consider low mobility users and a simple time-invariant strategy is sufficient to achieve promising results as demonstrated in \cite{15}. However, when user mobility increases, the time-invariant strategy often leads to mismatched traffic demands and transmission capabilities, and the link utilization ratio is limited. Meanwhile, conventional traffic splitting algorithms require the computational complexities to grow exponentially with the number of available bands, which may not be suitable for practical implementations of 5G networks with 16 component carriers as well.

In this paper, we consider a novel traffic splitting mechanism for multi-stream CA in hybrid sub-6 GHz and mmWave bands scenario. In order to address the above issues, we model different propagation losses for different frequency bands and propose a low-complexity traffic splitting algorithm based on fuzzy proportional integral derivative (PID) control mechanism, where the main contributions are listed below. 
\begin{itemize}
\item {\em Reduced Feedback Overhead.} Different from the conventional feedback based traffic splitting mechanisms, our proposed approach relies on observing the local RLC buffer statuses of component carriers to approximate UEs' behaviors, which is more favorable for the practical deployment.
\item {\em Low-complexity Implementation.} To reduce the implementation complexity, we approximate the original time-varying mixed-integer optimization problem with short term expectation maximization. Meanwhile, we utilize the fuzzy control-based PID adaptation instead of reinforcement learning scheme to achieve lower complexity and quicker convergence.
\end{itemize}

Since the proposed algorithm only relies on the local RLC buffer information of sub-6 GHz and mmWave bands and minimizes frequent feedback from user equipment (UE) side, it can also be easily extended to machine-to-machine \cite{RLSTM,xu2020deep}, or vehicular-to-vehicular communications \cite{RDDPG}. Through some numerical examples, our proposed traffic splitting mechanism can achieve more than 90\% link resource utilization ratio for different UE transmission requirements with different mobilities, which corresponds to 10\% throughput improvement, if compared with several conventional baseline schemes, such as \cite{15} and \cite{16}.

The remainder of the paper is organized as follows. In Section~\ref{sect:sys} we describe the high-layer splitting model for multi-stream carrier aggregation of 5G non-standalone network, in Section~\ref{sect:des} we discuss the design and deployment of allocation mechanism. In  Section~\ref{sect:exper} we report some examples and results, and we conclude the paper and provide insights on future works in Section~\ref{sect:con}. 
\begin{table*}[h]
\caption{Notation and Acronym}
\label{table:Notation}
\begin{tabular}{p{4cm}|p{12cm}}
 \hline 
\emph{Notation and Acronym}  & \emph{Definition} 
   \\  \hline  
$N_{SCC}$, $S$, $s^{th}$& Number of SCCs, set of SCCs and the $s^{th}$ SCCs, respectively  
\\ 
$Q^{P}(t)$ & Set of data packets buffered in the PDCP layer   
\\ 
$Q^{P}_{i}(t)$, $\vert Q^{P}_{i}(t)\vert$ &  Set and quantity of arriving packets from the SDAP layer   
\\ 
$Q^{P}_{P/s,o}(t)$, $\vert Q^{P}_{P/s,o}(t)\vert$              & Set and quantity of data packets departed to the PCC and the $s^{th}$ SCC
\\ 
$N^{P}_{\max}(t)$, $N^{P}_{\min}(t)$               &  Maximum and minimum packet indication in $Q^{P}(t)$, respectively
\\ 
$A_{P/s}(t)$             & Transmission strategy for the the PCC and the $s^{th}$ SCC, respectively
\\ 
$Q_{P/s}^R(t)$, $\vert Q_{P/s}^R(t)\vert $               & Set and quantity of data packets in the PCC and the $s^{th}$ SCC RLC buffer
\\ 
$ S_{P/s}^M(t)$, $\vert  S_{P/s}^M(t)\vert $               &  Set and quantity of data packets transmitted in MAC and PHY layer of the PCC and the $s^{th}$ SCC
\\ 
$\rho_{P/s}$            & Normalization factors of the PCC and the $s^{th}$ SCC, respectively
\\ 
$\gamma_{P/s}(t)$             & Signal-to-interference-and noise ratio (SINR) of the PCC and $s^{th}$ SCC, respectively
\\ 
$N_{th}$                &  Normalized threshold for successful packet delivery
\\ 
 $\vert Q_U^{P}(t)\vert$              &  Quantity of data packets successfully received at the UE side
\\ 
$T$, $L$, $N$               &  Duration of transmission time slots, the data packets generated from IP flows and the length of prediction period, respectively
\\ 
$B(t)$               & $B(t)=\vert Q_P^R(t)\vert-\sum_{s=1}^{N_{SCC}}\vert Q_s^R(t)\vert$
\\ 
$K_p(t)$, $K_i(t)$, $K_d(t)$               & Time varying coefficients of proportion,integration and derivation, respectively
\\
\emph{CA}                &  Carrier Aggregation
\\ 
\emph{PCC}                & Primary Component Carrier
\\ 
\emph{SCC}                & Secondary Component Carrier
\\ 
\emph{SDAP}                &  Service Data Adaptation Protocol
\\ 
\emph{PDCP}                &  Packet Data Convergence Protocol
\\ 
\emph{RLC}                & Radio Link Control
\\ 
\emph{MAC}                &  Medium Access Control
\\ 
\emph{PHY}                &  Physical
\\ 
\emph{PID}                & Proportional Integral Derivative
 \\ \hline
\end{tabular}
\end{table*}

\section{System Model} \label{sect:sys}
\begin{figure*}
\centering
\includegraphics[width = 165mm, height =80 mm]{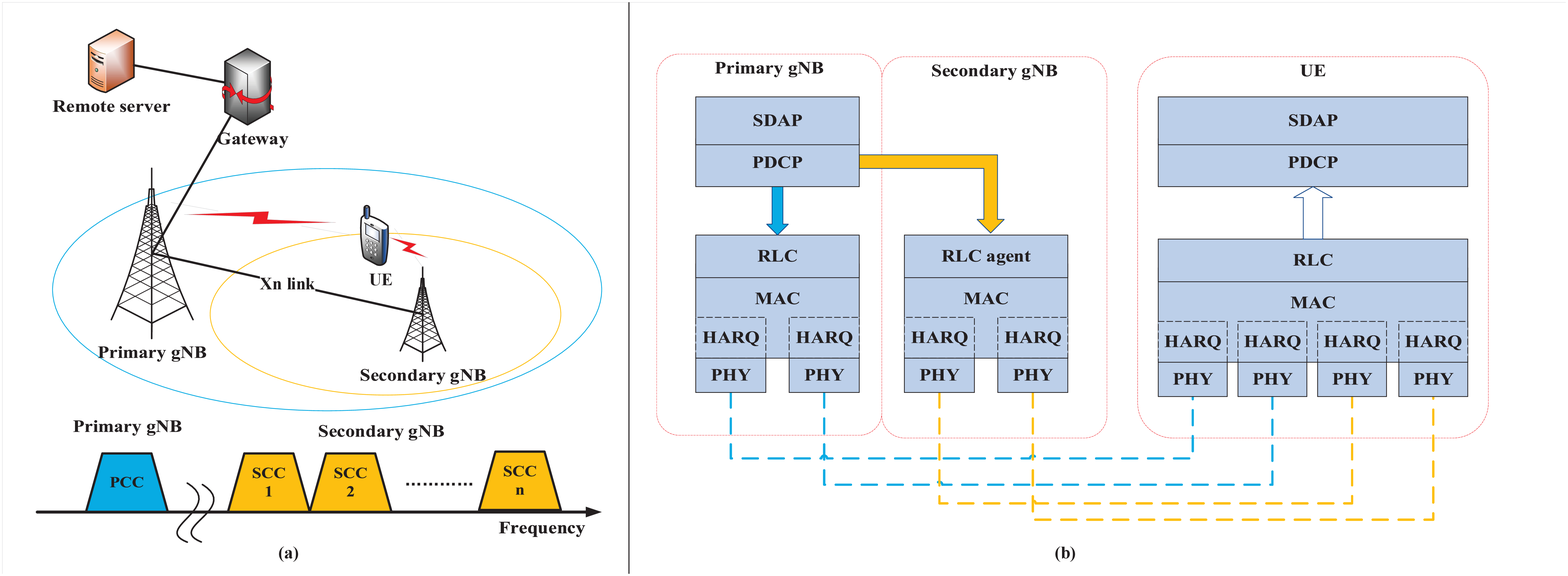}
\caption{Diagram of carrier aggregation across different cells. }
\label{fig:system}
\end{figure*}
Consider an inter-band non-contiguous CA system as shown in Fig.~\ref{fig:system}, where a user equipment (UE) connects a primary 5G new radio BS (gNB) and a secondary gNB simultaneously. The primary gNB transmits on a sub-6 GHz PCC and the secondary gNB can deliver information on $N_{SCC}$ SCCs (denoted by $\mathcal{S}$) in mmWave bands. Two gNBs can communicate with each other via the Xn link as defined in \cite{17} and the associated delay is simply normalized to $d_{X_n}$. Based on the 5G U-plane protocol stack \cite{18}, an inter-band non-contiguous CA transmission contains the following procedures as illustrated below.\footnote{ In our formulation, we do not consider any specific constraints on sub-6 GHz and mmWave bands. but in the numerical evaluation, we adopt the practical values of sub-6 GHz and mmWave bands to obtain an insightful results of deployed 5G networks.} 

\begin{itemize}
 \item{\em Service Data Adaptation Protocol (SDAP):} The main target of SDAP layer is to map different QoS requirements to data radio bearers (DRBs). For illustration purposes, we assume a simple transparent transmission policy is adopted and the lower layers can directly receive $L$ data packets generated from higher layers in each time slot.
 \item{\em Packet Data Convergence Protocol (PDCP):} Denote 
 $Q^{P}(t)=\{N^{P}_{\min}(t), N^{P}_{\min}(t)+1,\ldots, N^{P}_{\max}(t)\}$ to the buffered packets of PDCP layer at the time slot $t$. We can have the status update expressions as follows.\footnote{ According to 3GPP Release 15 specfication \cite{Nstand}, the SDAP layer is implemented in PCC only to guarantee different QoS requirements to DRBs.}
 \begin{eqnarray}
 Q^{P}(t+1) & = & Q^{P}(t) \cup Q^{P}_{i}(t) - Q^{P}_{P,o}(t) \nonumber \\ 
 && - \cup_{s \in \mathcal{S}} Q^{P}_{s,o}(t), \label{eqn:PDCP_Q}
 \end{eqnarray}
 where $Q^{P}_{i}(t)$ denotes the arrival packets from upper layers, and $Q^{P}_{P/s,o}(t)$ denotes the departure packets to the PCC or the $s^{th}$ SCC, respectively. The maximum and minimum packet indices in $Q^{P}(t+1)$ are updated by,
 \begin{eqnarray}
 N^{P}_{\max}(t+1)&= N^{P}_{\max}(t) + |Q^{P}_{i}(t)|, \\
 N^{P}_{\min}(t+1)&= N^{P}_{\min}(t) + |Q^{P}_{P,o}(t)|\nonumber\\&+ |\cup_{s \in \mathcal{S}} Q^{P}_{s,o}(t)|,
 \end{eqnarray}
 where $|\cdot|$ denotes the cardinality of the inner set, $|Q^{P}_{i}(t)| = L$, and $|\mathcal{S}| = N_{SCC}$.
 Denote $A_{P/s}(t)$ to be the transmission strategy for the PCC and $s^{th}$ SCC at the time slot $t$ in the  traffic splitting mechanism. $Q^{P}_{P/s,o}(t)$ can be updated via the following expressions.
\begin{eqnarray}
\vert Q^{P}_{P/s,o}(t+1)\vert &=\vert Q^{P}_{P/s,o}(t)\vert+A_{P/s}(t), \label{eqn:P_OUT}
\end{eqnarray}
where $A_{P/s}(t) = 1$ indicates the buffered packet of PDCP layer is successfully transmitted to the PCC or $s^{th}$ SCC, and $A_{P/s}(t) = 0$ otherwise.
 \item{\em Radio Link Control (RLC):} This layer receives departure packets of the PDCP layer and sends the processed packets to the UE side via lower layers.  Therefore, the RLC buffer status of the PCC and $s^{th}$ SCC, e.g., $Q_{P/s}^R(t+1)$, can be updated via the following relations.
 \begin{eqnarray}
 Q_{P/s}^R(t+1)&=&Q_{P/s}^R(t)  \cup  Q_{P/s,o}^P(t)\nonumber\\&&- S_{P/s}^M(t), \label{eqn:RLC_Q}
\end{eqnarray}
where $S_P^M(t)$ and $S_s^M(t)$ denote the processing capabilities of lower layers.

\item{\em Medium Access Control (MAC) \& Physical (PHY):} In this layer, we use the abstracted model of MAC and PHY layers and the MAC and PHY layer transmission rates of PCC and several SCCs are given by, 
\begin{eqnarray}
\vert S_{P}^M(t)\vert&=& \left\{
\begin{aligned}
&\lfloor \rho_{P}/\rho_{s} \rfloor  ,~~~~~~~\textrm{if}
\\&~~\rho_{s} \log_2(1+\gamma_{P}(t))\geq N_{th}, \label{eqn:M_P}\\
&0  , ~~~~~~~ \textrm{otherwise}.
\end{aligned}
\right.\\
\vert S_{s}^M(t)\vert& =& \left\{
\begin{aligned}
1 & , ~~~~~~~~~~~~~~~~~~~ \textrm{if} \\&\rho_{s}\log_2(1+\gamma_{s}(t))\geq N_{th},\label{eqn:M_s} \\
0 & , ~~~~~~~ \textrm{otherwise}.
\end{aligned}
\right.
\end{eqnarray}
In the above equations, $\rho_{P/s}$ are the normalization factors, which include the effects of bandwidth, payloads and transmission durations. $\gamma_{P/s}(t)$ denote the signal-to-interference-and noise ratio (SINR) of the PCC and $s^{th}$ SCC, respectively. $N_{th}$ represents the normalized threshold for successful packet delivery. In the practical deployment, we often utilize sub-6 GHz band for PCC and mmWave bands for SCCs, and the corresponding SINR expressions are given by \cite{LTEstudy,2016study},
\begin{eqnarray}
\gamma_{P}(t)&=&PT_P-10\cdot\log_{10}(h_P\cdot\alpha_P(t))\label{eqn:sub_6G}\\
\gamma_{s}(t)&=&PT_s-10\cdot\log_{10}(h_s\cdot\alpha_s(t))\label{eqn:mmwave}
\end{eqnarray}
where $PT_{P/s}$ denotes the transmission power of the PCC and $s^{th}$ SCC, and $h_{P/s}$ are the normalized path losses for PCC-UE and the $s^{th}$ SCC-UE links. $\alpha_{P/s}(t)$ are the time-varying fading coefficients, where $\alpha_{P}(t)$ follows a Rice distribution with unit mean and variance $\sigma_{P}^2$, and $\alpha_{s}(t)$ follows a Rayleigh distribution with unit mean and variance $\sigma_s^2$, as specified in \cite{Rice,Ray}.
\end{itemize}

With the CA transmission strategy and the abstracted MAC and PHY layer models, the target UE collects all the data packets from different component carriers in the RLC layer and sends them to the PDCP layer according to \cite{19}. Finally, the UE received packets in the PDCP layer can be modeled as follows.

\begin{eqnarray}
  \vert Q_U^{P}(t) \vert&=&\vert S_{P}^M(t)\vert+\sum_{s=1}^{N_{SCC}}
 \vert S_{s}^M(t)\vert, \label{eqn:f_func}
\end{eqnarray}
where $  \vert Q_U^{P}(t)\vert$ indicates the number of data packets successfully received at the UE side.

The following assumptions are adopted throughout the rest of this paper. First, the data processing among different layers is assumed to be zero delays and error-free. Second, the packet lengths of different layers are assumed to be fixed for simplicity and the headers of different layers are considered to be negligible. Third, infinite buffer sizes are assumed for different layers, such that the buffer overflow effect is not considered. Moreover, the RLC layer works in the acknowledgment mode according to \cite{20}. For illustration purposes, we summarize all the notations and acronyms in Table~\ref{table:Notation}.

\section{Problem Formulation} \label{sect:Pro}

In this section, we formulate the transmission duration minimization problem based on the above CA transmission model. In order to adapt with the wireless fading environment, we consider a dynamic packet transmission strategy $A_{P/s}(t)$ in the PDCP layer. With the exact expressions of $Q^{P}_{P/s,o}(t)$, the transmission duration $T$ is thus determined by accumulating the value of $ \vert Q_U^P(t)\vert$, and the optimal packet transmission strategy in the  traffic splitting mechanism can be obtained by solving the following minimization problem. 
\begin{Prob}
[Original Problem] \label{prob:ori} The optimal transmission duration can be achieved by solving the following optimization problem.
\begin{eqnarray}
\underset{\{A_{P/s}(t)\}}{\textrm{minimize}}  && T, \label{eqm:ori_obj}\\
\textrm{subject to} 
                && \textcolor{red}{\eqref{eqn:PDCP_Q}-\eqref{eqn:f_func}},\ A_{P/s}(t)\in\{0, 1\},\notag \\
                &&\sum_{t=0}^{T}\vert Q_U^{P}(t)\vert\geq L ,\label{eqn:cons2} \\
                &&  Q^{P}_{P,o}(t) \cup \{\cup_{s \in \mathcal{S}} Q^{P}_{s,o}(t)\}\subset  Q^P(t), \label{eqn:cons3} \nonumber\\
                &&\\
                && Q^{P}_{i,o}(t) \cap Q^{P}_{j,o}(t)= \emptyset, \notag \\
                 && \forall \ i,j \in \{P\} \cup \mathcal{S}, \textrm{and} \ i \neq j 
                \label{eqn:cons4}.
\end{eqnarray}
where \eqref{eqn:cons3} and \eqref{eqn:cons4} guarantee that all delivered packets to PCC and SCCs RLC layers are taken from the current PDCP buffer and do not overlap with each other.
\end{Prob}

The above problem is in generally difficult to solve due to the following reasons. First, the searching space of packet transmission strategies grows exponentially concerning the sizes of $Q^{P}(t)$ and the optimal strategy is a typical mixed-integer optimization problem (MIOP). Second, due to the time-varying wireless conditions, e.g., $\gamma_P(t)$ and $\gamma_s(t)$, the searching spaces of $Q_{P/s,o}^P(t)$ are dynamically changing, which further increases the searching complexity as well. To make it mathematically tractable, we focus on the following approximated dynamic programming problem, where the instantaneous transmission action, e.g.,$\{A_{P/s}(t)\}$, can be determined by the following maximization problem.
\begin{Prob}[Approximated Problem] \label{prob:equ}
For any given time slot $t^{\prime}\in \left[0,T\right]$, the optimal packet allocation action can be determined as follows.
\begin{eqnarray}
\underset{A_{P/s}(t^{\prime})}{\textrm{maximize}}  &&  \sum_{t=t^{\prime}}^{t^{\prime}+N}\mathbb{E}\left[ \vert Q_U^{P}(t)\vert\big| \{Q_{P/s}^R(t^{\prime})\}, A_{P/s}(t^{\prime}) \right], \nonumber\label{eqn:equ}\\\\
\textrm{subject to} 
 && \eqref{eqn:M_P}-\eqref{eqn:mmwave}, \eqref{eqn:cons3}-\eqref{eqn:cons4}, \nonumber \\
                &&\vert Q_{P/s}^R(t^{\prime})\vert \in \{0,1,...,L\} \label{eqn:cons5}.
\end{eqnarray}
where $N \ll T$ denotes the length of prediction period in the near future.
\end{Prob}
\begin{theorem}\label{thm:approx}
Problem \ref{prob:ori} can be well approximated by recursively solving Problem \ref{prob:equ}.
\end{theorem}
\begin{proof}
Please refer to Appendix~\ref{appd:thm1} for the proof.
\end{proof}

By applying Theorem~\ref{thm:approx}, the original transmission duration minimization problem has been modified to maximize the expected number of successfully received packets in a known state, where conventional single-slot greedy based algorithms can be used \cite{21}. Since the evaluation of $  \vert Q_U^{P}(t)\vert$ highly relies on the time-varying variables $S_{P/s}^M(t)$ and ideal feedback scheme from the target UE, greedy based algorithms can hardly be implemented in practice. To deal with that, we realize that the RLC layer buffer status can be utilized to model the behaviors of $S_{P/s}^M(t)$, and with some mathematical manipulations as specified in Appendix~\ref{appd:thm2}, we have the following
simplified version.
\begin{Prob}[Simplified Problem]
For any given time slot $t^{\prime}\in \left[0,T\right]$, the simplified transmission strategies can be determined as follows.
\label{prob:rec}
\begin{eqnarray}
\underset{A_{P/s}(t^{\prime})}{\textrm{minimize}}  && \big\vert B(t^{\prime})+(N+1)\cdot\big[A_{P}(t^{\prime})-\notag\\&&\sum_{s=1}^{N_{SCC}}A_{s}(t^{\prime})-(\lfloor\rho_{P}/\rho_{s}\rfloor-N_{SCC})\big]\big\vert,\notag\\\\
\textrm{subject to} 
 && \eqref{eqn:M_P}-\eqref{eqn:mmwave}, \eqref{eqn:cons5}, t^{\prime}\in \left[0,T\right],\nonumber\\
 && B(t^{\prime})=\vert Q_P^R(t^{\prime})\vert-\sum_{s=1}^{N_{SCC}}\vert Q_s^R(t^{\prime})\vert. \label{eqn:b(t)}
\end{eqnarray}
\end{Prob}

Based on Problem \ref{prob:rec}, we can derive several low-complexity algorithms as shown later. This is because all the changing variables in Problem \ref{prob:rec} can be locally obtained without any interactions with the terminal side.
\begin{figure}
\centering
 \includegraphics[width = 3.3in]{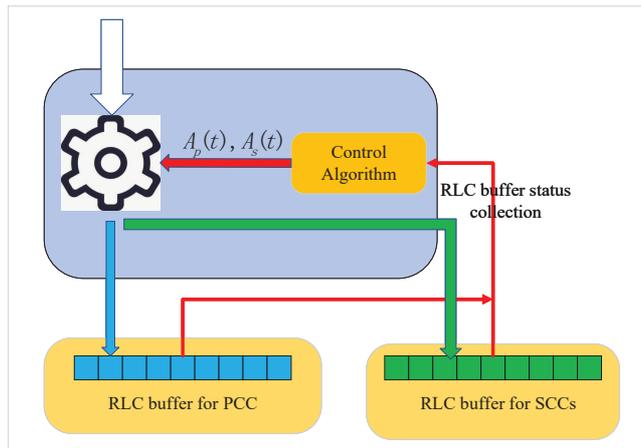}
\caption{The schematic diagram of traffic splitting mechanism based on RLC buffer status.}
\label{fig:Mec}
\end{figure}
\section{Proposed Traffic Splitting Mechanism} \label{sect:des}
In this section, instead of solving the above optimization problem using the brute force approach, we propose a low complexity dynamic traffic splitting mechanism using fuzzy logic control structure \cite{22}. As shown in Fig.~\ref{fig:Mec}, the proposed traffic splitting mechanism\footnote{ In our formulation, we adopt an abstract general model, in order to develop a generic algorithm. 
} collects RLC buffer information to obtain $B(t')$, and then determine the dynamic packet transmission strategy $A_{P/s}(t')$ by solving Problem \ref{prob:rec}.

As shown in Algorithm~\ref{Alg:Algorithm}, the proposed fuzzy logic control-based algorithm can be divided into the following two stages, e.g., initialization and adaptation.

\begin{itemize}
    \item{\em Stage 1 (Initialization)}: In the initialization stage when $t^{\prime}\leq N$, we simply let $A_{P/s}(t') = 1$, since there is limited information about the packet transmission process.
    \item{\em Stage 2 (Adaptation)}: In the adaptation stage when $t^{\prime} > N$, we first obtain $B(t')$ according to \eqref{eqn:b(t)}, and then determine the transmission mode according to the value of $B(t')\cdot B(t'-1)$.
    
    {\em Static Mode $\left(B(t')\cdot B(t'-1) > 0\right)$}: In the static mode, we follow the previous packet transmission strategy and simply choose the current actions $A_{P/s}(t')$ as,
    \begin{eqnarray}
    A_{P/s}(t') = A_{P/s}(t'-1), \forall t'.
    \end{eqnarray}
    
    {\em Dynamic Mode $\left(B(t')\cdot B(t'-1) \leq 0\right)$}: In the dynamic mode, the current transmission actions are determined from the previous actions in the last $N$ time slots. For illustration purposes, we define $k(t')$ as the auxiliary variable, where the mathematical expression is given as, 
\begin{eqnarray}
k(t')&=& \left\lfloor\frac{\sum_{s=1}^{N_{SCC}}\sum_{n=1}^{N}A_s(t'-n)}{\sum_{n=1}^{N}A_P(t'-n)} \right\rfloor. \label{eqn:k_func}
\end{eqnarray}

\floatname{algorithm}{Algorithm}   
\renewcommand{\algorithmicrequire}{\textbf{input:}}  
\renewcommand{\algorithmicensure}{\textbf{output:}}  
\begin{algorithm}[h]
          \caption{Proposed Fuzzy Logic Control Based Algorithm} 
\label{Alg:Algorithm}
          
\begin{algorithmic}[1] %每行显示行号  
            \REQUIRE the RLC buffer status $\vert Q_P^R(t^{\prime})\vert$, $\vert Q_s^R(t^{\prime})\vert$; 
            \ENSURE  $A_P(t^{\prime})$,$A_s(t^{\prime})$ 
            \STATE Stage 1:
               \IF{$t^{\prime}\leq N$}
             \STATE $A_{P/s}(t^{\prime}) = 1$;
             \ELSE
              \STATE  go to Stage 2 
            \ENDIF  
              \STATE Stage 2:
               \STATE obtain $B(t^{\prime})$ according to \eqref{eqn:b(t)}
               \IF{$\left(B(t^{\prime})\cdot B(t^{\prime}-1) > 0\right)$}
             \STATE $ A_{P/s}(t^{\prime}) = A_{P/s}(t^{\prime}-1)$
             \ELSE
             \IF{$(t' mod N==0)$}
              \STATE update $K_p(t')$,$K_i(t')$,$K_d(t')$ according to 
               \STATE equation \eqref{eqn:K_update}
              \ELSE
             \STATE  $K_p(t')=K_p(t'-1)$, $K_i(t')=K_i(t'-1)$, 
             \STATE  $K_d(t')=K_d(t'-1)$
               \ENDIF 
             \STATE update $A_P(t^{\prime})$ according to equation \eqref{eqn:A_P}
             \STATE update $A_s(t^{\prime})$ according to equation \eqref{eqn:A_s}
            \ENDIF  
             \STATE output $A_P(t^{\prime})$,$A_s(t^{\prime})$;
            
 \end{algorithmic}  
\end{algorithm}
To enable a PID based control strategy, we form a second-order filtering algorithm to obtain the incremental value \cite{23}, $G(t')$, e.g.
\begin{eqnarray}
G(t')&= &K_p(t')\cdot\left[B(t^{\prime})-B(t^{\prime}-1) \right]\nonumber\\&&+K_i(t')\cdot B(t^{\prime})+K_d(t')\nonumber\\&&\cdot\left[B(t')-2B(t'-1)+B(t'-2) \right]\label{eqn:PID}.\nonumber\\
\end{eqnarray}
where $K_p(t')$, $K_i(t')$, and $K_d(t')$ represent the time varying coefficients of proportion, integration and derivation, respectively.
With the above calculated variable $k(t^{\prime})$ and $G(t^{\prime})$, we have the current transmission actions as, 
\begin{eqnarray}
 &&A_{P}(t')= \left\{
\begin{aligned}
&\delta\left(t'\bmod N-i\cdot k(t')\right), \textrm{if} \\
&~~~t'\bmod N \leq G(t')\cdot k(t').\label{eqn:A_P}\\
&\delta\left(t'\bmod N-j\cdot\left[ k(t')+1\right]\right),\\
&~~~~~~~\textrm{otherwise}.
\end{aligned}
\right.\\
&&A_{s}(t')=1-A_{P}(t'),\forall s \in S, \label{eqn:A_s}
\end{eqnarray}
where $i\in [1,G(t^{\prime})]$, $j\in[1,N-G(t^{\prime})]$, and $\delta(\cdot)$ is the unit-impulse function as defined in \cite{24}.
\end{itemize}

The convergence property of the above proposed algorithm is greatly affected by the values of $K_p(t')$, $K_i(t')$, and $K_d(t')$ as proved in \cite{25}. In order to adapt to different application scenarios, we use fuzzy control-based solution to dynamically adjust the above control parameters of PID \cite{Fuzzy_Prove}, which is able to achieve the desired performance with unknown nonlinearities, processing delays, and disturbances. As shown in the Fig.~\ref{fig:PID}, the fuzzy control-based PID parameter optimization consists of three modules, namely fuzzifier, fuzzy inference, and defuzzfier. 

In the fuzzifier module, we first normalize the value of $B(t')$ and $B(t') - B(t'-1)$ by the maximum value $B_{max}$, and then calculate the corresponding membership degrees $\mathcal{D}_{B}(t')$ and $\mathcal{D}_{E}(t')$ according to the triangular membership function $M(\cdot)$ \cite{26}, respectively.
\begin{eqnarray}
\mathcal{D}_{B}(t') & = & M\left( \frac{B(t')}{B_{\max}} \right), \\
\mathcal{D}_{E}(t') & = & M\left( \frac{B(t') - B(t'-1)}{2 \times B_{\max}} \right).
\end{eqnarray}
In the fuzzy inference module, three fuzzy rule tables, including $\mathcal{T}_p(\cdot)$, $\mathcal{T}_i(\cdot)$, and $\mathcal{T}_d(\cdot)$ as defined in \cite{27} are applied to update $K_p(t')$, $K_i(t')$ and $K_d(t')$, respectively. In the defuzzifier module, the incremental output values of $K_p(t')$, $K_i(t')$ and $K_d(t')$ are obtained through,
\begin{eqnarray}
 \left\{
\begin{aligned}
&K_{p}(t')=K_{p}(t'-1)+\Delta K_{p}(t'),\label{eqn:K_update}\\
&K_{i}(t')=K_{i}(t'-1)+\Delta K_{i}(t'),\\
&K_{d}(t')=K_{d}(t'-1)+\Delta K_{d}(t'),
\end{aligned}
\right.
\end{eqnarray}
where $\Delta K_{p/i/d}(t^{\prime})$ can be calculated via
\begin{eqnarray}
&&\Delta K_{p/i/d}(t^{\prime})=\begin{bmatrix}\mathcal{D}_{B}(t')&1-\mathcal{D}_{B}(t')\end{bmatrix}
\cdot \mathcal{T}_{p/i/d}\nonumber\\&&\cdot\left(\begin{tiny}\begin{bmatrix} \mathcal{D}_{B}(t')\cdot \mathcal{D}_{E}(t')& \mathcal{D}_{B}(t')\cdot(1-\mathcal{D}_{E}(t'))\\
(1-\mathcal{D}_{B}(t'))\cdot \mathcal{D}_{E}(t')&(1-\mathcal{D}_{B}(t'))\cdot (1-\mathcal{D}_{E}(t')) \end{bmatrix}\end{tiny}\right)\notag\\
&&\cdot\begin{bmatrix}\mathcal{D}_{E}(t')\\1-\mathcal{D}_{E}(t')\end{bmatrix}.
\end{eqnarray}

\begin{figure}
\centering
 \includegraphics[width = 3.3in]{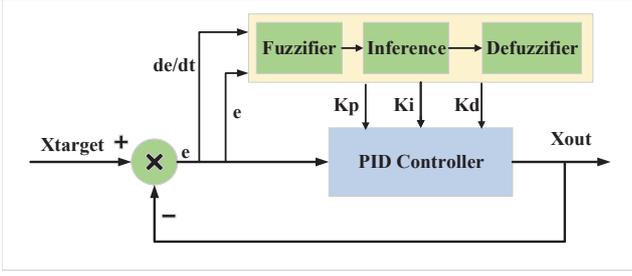}
\caption{  Fuzzy-PID structure  }
\label{fig:PID}
\end{figure}
\begin{figure}[h]
\centering
 \includegraphics[width = 3.3in]{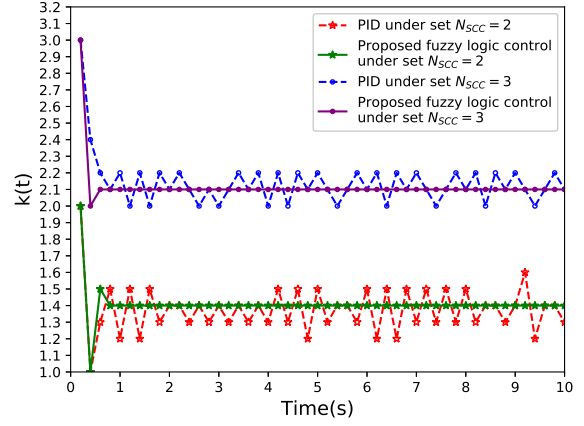}
\caption{The convergence property of PID control and the proposed fuzzy logic control under the number of $N_{SCC}$ setting.}
\label{fig:split}
\end{figure}

The above  traffic splitting mechanism has the following advantages. First, we use a two-stage control algorithm to quickly approach the optimal strategy during the initialization period and keep the algorithm stable during the adaptation period. In the initialization period, we simply set all the actions to be active, e.g., $A_{P/s} = 1, \forall s$, to fulfill the buffer with the shortest period. In the adaptation period, we update the transmission actions based on a historial observation of $N$ time slots and the output value $G(t')$ of PID algorithm. Through this approach, we can quickly increase the number of transmission packets to explore the optimal transmission strategy and slightly adjust the transmission strategy according to the historical transmission strategy and buffer difference to maintain the stability of the algorithm. Second, we use the PID control algorithm with fuzzy based parameter optimization to guarantee the quick convergence property in different scenarios as proved in \cite{28}, \cite{29}. As shown in Fig.~\ref{fig:split}, the proposed fuzzy logic control based traffic splitting algorithm can quickly converge to the optimal value\footnote{ In the static user scenario with flat fading channel conditions, the optimal value should be equal to the ratio of link transmission capacities of PCC and $N_{SCC}$ SCCs as derived in \cite{GoodPUT}.}  with less than two rounds adaptation for both $N_{SCC} = 2$ and $N_{SCC} = 3$ cases.

\section{Experiment Results} \label{sect:exper}

In this section, we provide some numerical results to verify the proposed fuzzy logic control-based adaptive packet transmission mechanism. To provide a fair comparison, we use Network Simulator 3 (NS-3), currently implementing a wide range of protocols in C++ \cite{30}, with the most up-to-date 5G CA protocols as defined in \cite{31}. We simulate real network scenarios, using non-line-of-sight (NLOS) for communication in an urban macro fading condition \cite{2016study}, and other important simulation parameters are listed in Table~\ref{table:Parameters}. All the numerical simulations are performed on a Dell Latitude-7490 with i7-8650 CPU and 16GB memory. 
\def\tablename{Table}
\begin{table*}
\renewcommand\arraystretch{1.5}
\centering 
\caption{Simulation Parameters.}
	\begin{tabular}{c c c c}
	\hline \textbf{ Parameter } & \textbf{Value}&\textbf{ Parameter } & \textbf{Value} \\
\hline
Frequency of PCC& 4.9GHz &Bandwidth of PCC  &  100MHz   \\
\hline 
Frequency of SCC & 28GHz  &
Bandwidth of SCC &  100MHz      \\
\hline 
Transmission power of PCC & 28dBm  &
Transmission power of SCC &  35dBm      \\
\hline
TCP Congestion Control Algorithm & NewReno
&UE TCP Receive Window Size& 512KB\\
\hline
RLC Layer Transport Mode   & AM&
RLC layer Polling PDU Threshold  & 100\\
\hline
Xn link delay& 2ms&
Xn link data rate & 1Gbps\\
\hline 
 Variance  $\sigma_P^2$ & 0.0004  &
Variance $\sigma_s^2$ & 0.27      \\
\hline 
\end{tabular}
\label{table:Parameters}
\end{table*}

In order to provide a more intuitive result, we define the link resource utilization ratio, $\eta$, to be the performance measure, which is given by,
\begin{eqnarray}
    \eta=\frac{\sum_{t=1}^T \vert Q_U^P(t)\vert}{\sum_{t_p=1}^T \vert Q_U^P(t_p)\vert+\sum_{t_S=1}^T \vert Q_U^P(t_S)\vert}\times 100\%.
    \label{eqn:eta}
\end{eqnarray}
In the above expression, $t_p$ and $t_S$ denotes the time indexes of PCC transmission only and $N_{SCC}$ SCCs transmission only modes, respectively. The physical interpretation is the ratio of average end-to-end PDCP throughput of CA transmission over the maximum end-to-end PDCP transmission throughput provided by PCC and $N_{SCC}$ SCC links. 

In the following numerical examples, we test the average end-to-end PDCP throughput and the link resource utilization ratio, $\eta$, with several baselines. {\em Baseline 1 (BWA)} \cite{15}: The packet transmission strategy is to allocate the PDCP packets according to the available bandwidths of PCC and $N_{SCC}$ SCCs. {\em Baseline 2 (LTR)}  \cite{16}: The packet transmission strategy is to allocate the PDCP packets according to the measured end-to-end link delay. {\em Baseline 3 (No Fuzzy):} The packet transmission strategy is similar to our proposed mechanism except that the PID control parameters are not optimized using fuzzy processes. {\em Baseline 4 (Q-learning)} \cite{Q-learn}: The packet transmission strategy is to allocate the PDCP packets based conventional reinforcement learning based approach as proposed in \cite{de2019dm}.

%实验1.以n取2验证不同分流比下buffer差值以及吞吐量大小
%实验2,以n取2分流算法部署吞吐量随变化和
%实验3,移动场景实验

%buffer difference and throughput
\subsection{Buffer Status versus Average Throughput}

%the buffer difference $B(t)$ versus
%we plot the average end-to-end throughput and the buffer difference $B(t)$ versus different transmission strategy.
%we rewrite equation \eqref{eqn:k_func}
In Appendix~\ref{appd:thm2}, we derive that the buffer difference $\vert B(t)\vert$ is negatively correlated with the throughput of PDCP layer that we can obtain the maximum throughput by minimizing the buffer difference $\vert B(t)\vert$. To numerically demonstrate the above transformation, we plot the average end-to-end throughput and the buffer difference $\vert B(t)\vert$ versus different transmission strategies and different numbers of $N_{SCC}$ setting. 
%To provide a more straightforward view of transmission strategy, we rewrite equation \eqref{eqn:k_func} and let $k(t)$ retain one decimal place to represent different transmission strategies.
% calculate resource link utilization, and  to observe the change of the buffer difference more intuitively, we always take the absolute value of buffer difference.

\begin{figure}
\centering
 \includegraphics[width = 3.3in]{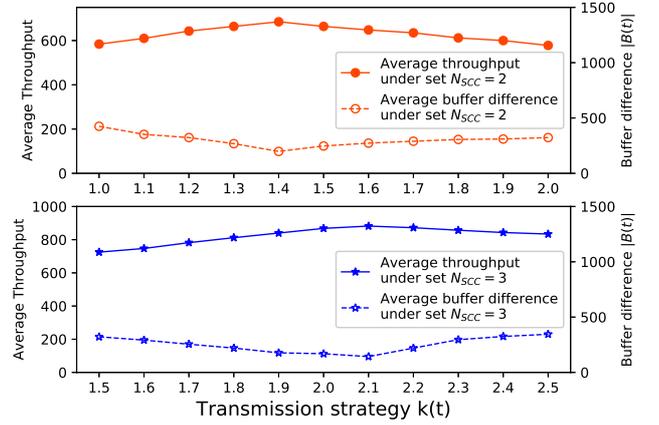}
\caption{Numerical results on throughput and average buffer difference $\vert B(t)\vert$ under different transmission strategies $k(t)$ and the number of $N_{SCC}$ setting.}
\label{fig:buffer_differ}
\end{figure}
The experimental results are shown in Fig.~\ref{fig:buffer_differ}, the buffer difference $\vert B(t)\vert$ decreases when the PDCP throughput gradually increases, and vice versa. Only when the buffer difference reaches the minimum value, the PDCP throughput reaches the maximum. Therefore, the negative correlation between buffer difference $\vert B(t)\vert$ and throughput has been verified that we can adjust the transmission strategy through the buffer difference to obtain the maximum throughput.

\subsection{Static User Scenario}
In the static user scenario, the receiving UE for multi-stream CA transmission remains static, which is located $100$ meters away from the primary gNB during the entire transmission period. Without the UE mobility, the long-term channel statistics remain stable in this case, and we plot the link resource utilization ratio versus $N_{SCC}$ in Fig.~\ref{fig:stationary} to demonstrate the benefits of the proposed adaptive traffic splitting mechanism.

As shown in Fig.~\ref{fig:stationary}, the proposed adaptive traffic splitting scheme outperforms all four conventional baselines under different $N_{SCC}$ settings. For baselines 1 to 4, the achievable link resource utilization ratios are between 80\% to 92\%, while our proposed scheme can reach as much as 95\%. Meanwhile, since the entire bandwidth of SCCs is much greater than PCC as the number of SCCs increases, the advantages of the proposed splitting scheme compared to four baselines in terms of the link resource utilization ratio is decreasing. 

This is due to the following three reasons. First, by comparing with baseline 1, the proposed splitting scheme allows to dynamically adjust the packet transmission strategy for PCC and SCCs, which is more robust for dynamic channel variations. Second, by comparing with baseline 2 and baseline 4, the proposed splitting scheme estimates the transmission capability without explicit feedback from UEs and the end-to-end throughput evaluation as a reward, which saves the transmission bandwidth for delay feedback and the computational complexity for policy evaluation. Third, by comparing with baseline 3, the proposed splitting scheme dynamically optimizes the PID control parameters to obtain an additional 2-3\% improvement of the link resource utilization ratio.

\begin{figure}
\centering
 \includegraphics[width = 3.3in]{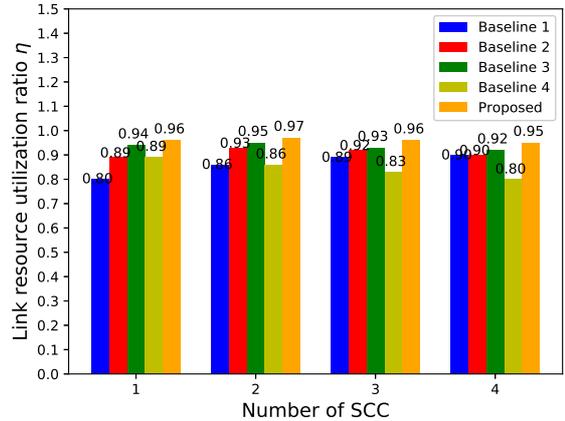}
\caption{Bar graph of link resource utilization ratio of static users with different carrier numbers under different traffic splitting mechanism deployments.}
\label{fig:stationary}
\end{figure}

\begin{figure}[h]
\centering
 \includegraphics[width = 3.3in]{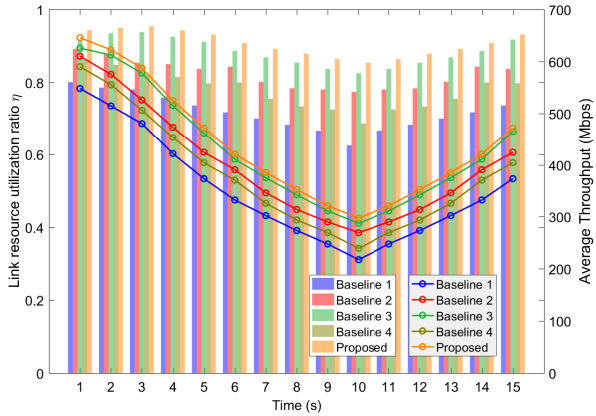}
\caption{In mobile user scenario, the change curve of the
link resource utilization ratio $\eta$ and average throughput over time under different traffic splitting mechanism deployments.}
\label{fig:move}
\end{figure}
\subsection{Mobile User Scenario}
In the mobile user scenario, the receiving UE for multi-stream CA transmission is moving away from the primary gNB with a constant speed of $10 m/s$ and then getting back along the same path after 10 seconds. Different from the static user scenario, the end-to-end PDCP throughput will suffer from severe degradation due to the significant path loss when they are far apart. In this experiment, we keep the number of SCCs to be $N_{SCC} = 3$ and plot the link resource utilization ratio as well as the end-to-end PDCP throughput versus time in Fig.~\ref{fig:move}.

As shown in Fig.~\ref{fig:move}, the proposed adaptive traffic splitting scheme still outperforms all four baselines in terms of both the link resource utilization ratio and the end-to-end PDCP throughput as well. The link resource utilization ratios of four baselines are between 73\% to 89\%, while it reaches to 91\% for the proposed scheme, which corresponds to 3\% link resource utilization ratio improvement if compared with baseline 3 and more than 10\% if compared with other baselines. These numerical results further confirm that the proposed adaptive packet transmission scheme can quickly adapt to the drastically changing transmission capabilities.

\subsection{Complexity and Storage}
%实验设置，条件等介绍
%实验结果展示
%实验结果分析
% The proposed has lower time complexity and RAM utilization
% UE 和 terminal 同时出现
% 表格 性能指标单位（ms 和 %） complexity 和 storage
% static user / mobile user
In this experiment, we still keep the number of SCCs to be $N_{SCC} = 3$ and test the computational complexity and RAM usage of static and mobile user scenarios respectively.
As shown in Table~\ref{table:Compute}, the proposed adaptive traffic splitting mechanism has similar RAM usage as baseline 1 and 3, and the computational complexity of the proposed is slightly higher than that of baseline 1 and 3. This is because the fuzzy process requires additional optimization calculations. Baseline 2 shows the highest RAM usage with a large amount of data acquisition. Baselines 4 show the highest computational complexity, which is two times more than that of ours.
The proposed adaptive traffic splitting mechanism has lower complexity while ensuring optimal performance. This is because our mechanism is equivalent to the time-varying link capacity by observing the RLC buffer difference that the changing variables can be locally obtained without any interactions with the UE side.

\begin{table}[]
\caption{In the static and mobile user scenario, the computational complexity(ms) and RAM usage(\%) under the deployment of under different traffic splitting mechanism deployments.}
\begin{tabular}{cccc}
\hline
\textbf{Scenario}                                                      & \textbf{\begin{tabular}[c]{@{}c@{}} Splitting \\ Mechanism\end{tabular}} & \textbf{\begin{tabular}[c]{@{}c@{}} Computational\\ Complexity\end{tabular}} & \textbf{\begin{tabular}[c]{@{}c@{}}RAM\\ Usge\end{tabular}} \\ \hline
\multirow{4}{*}{\begin{tabular}[c]{@{}c@{}}Static\\ User\end{tabular}} & Baseline 1                                                                 & 3.6                                                                & 0.56                                                       \\
                                                                       & Baseline 2                                                                 & 8.1                                                                & 0.61                                                       \\
                                                                       & Baseline 3                                                                 & 3.7                                                                         & 0.56                                                                 \\
                                                                        & Baseline 4                                                                & 8.2                                                                         & 0.57                                                                 \\
                                                                       
                                                                       & Proposed                                                                   & 3.9                                                                       & 0.57                                                                  \\ \hline
\multirow{4}{*}{\begin{tabular}[c]{@{}c@{}}Mobile\\ User\end{tabular}} & Baseline 1                                                                 & 4.6                                                                         & 0.62                                                                  \\
                                                                       & Baseline 2                                                                 & 9.0                                                                         & 0.72                                                                \\
                                                                       & Baseline 3                                                                 & 4.5                                                                         & 0.62                                                                \\
                                                                             & Baseline 4                                                                & 10.2                                                                        & 0.63                                                                \\
                                                                       
                                                                       & Proposed                                                                   & 4.8                                                                         & 0.63                                                                \\ \hline
\end{tabular}
\label{table:Compute}
\end{table}

\section{Conclusion} \label{sect:con}
In this paper, we propose a low-complexity traffic splitting algorithm based on fuzzy PID in multi-stream CA scenario to better utilize the transmission capacities provided by sub-6 GHz and mmWave bands. Through end-to-end modeling of protocol stacks, our proposed traffic splitting algorithm is able to minimize the entire transmission duration with local RLC buffer information, which eventually improves the end-to-end PDCP throughput. Based on the numerical experiments, the proposed traffic splitting scheme can achieve more than 90\% link resource utilization ratio for both static and mobile user scenarios, and 50\% computational complexity reduction simultaneously. Through the above studies, we believe the proposed traffic splitting mechanism can be efficiently deployed in the practical 5G networks and achieve significant throughput improvement for multi-stream CA transmission with sub-6 GHz and mmWave bands. Meanwhile, since the proposed scheme does not rely on some specific constraints of mmWave bands, it can be easily extended to multi-stream carriers with different transmission rates.

\section*{Acknowledgement} \label{sect:Ack}
This work was supported by the National Natural Science Foundation of China (NSFC) under Grants 62071284,
61871262, 61901251 and 61904101, the National Key Research and Development Program of China under Grants
2019YFE0196600, the Innovation Program of Shanghai Municipal Science and Technology Commission under Grant
20JC1416400, Pudong New Area Science \& Technology Development Fund, and research funds from Shanghai Institute
for Advanced Communication and Data Science (SICS).

\theendnotes

% \section*{Appendix}
\begin{appendices}
% \appendices
\numberwithin{equation}{section}

\section{Appendix A}
% Proof of Theorem~\ref{thm:approx}
\label{appd:thm1}
Proof the Theorem~\ref{thm:approx}.

With limited packet length $L$, the problem of minimizing time $T$ can be transformed into maximizing the number of successfully received packets. Thus,
the objective function \eqref{eqm:ori_obj} in Problem~\ref{prob:ori} can be given as follows. 
\begin{eqnarray}
\underset{\{A_{P/s}(t)\}}{\textrm{maximize}}  &&  \sum_{t=0}^{T}\vert Q_U^P(t)\vert, \label{eqn:app_obj}
\end{eqnarray}
Without loss of generality, we use the standard Markov Decision Process (MDP) to represent the above problem. 

\textbf{Definition 1 (System State)}: Defined the state function $s(t)$ as RLC buffer status.
\begin{equation}
s(t) =\{ Q_P^R(t),Q_s^R(t)\},\forall s \in \mathcal{S}.
\end{equation}
where $\vert Q_{P/s}^R(t)\vert \in \{0,1,...,L\}$, the initial  state is $s(0) = (0,0,...,0)$. The state $s(t)$ is Markov’s which satisfies the following equation $P[s(t+1)|s(t)]=P[s(t+1)|s(0),...]$, $P[\cdot]$ is the state transition probability \cite{markov}.

\textbf{Definition 2 (Action)}:
The  action at the time $t$ is denoted
as $a(t)=A_{P/s}(t)$. It indicates whether the data packet is successfully transmitted to PCC or $s^{th}$ SCC RLC buffer.

\textbf{Definition 3 (Reward)}:
Defined the reward function $r(t)$ as the number of packets successfully received at the UE side. It is affected by current state $s(t)$ and action $A_{P/s}(t)$.
\begin{equation}
r(t)=  \vert Q_U^{P}(t)\vert.
\end{equation}
 Thus,
the objective function \eqref{eqm:ori_obj} in Problem~\ref{prob:ori} can be given as follows. 
\begin{eqnarray}
\underset{\pi=\{a(t)\}}{\textrm{maximize}}  &&  \sum_{t=0}^{T}\mathbb{E}\left[r(t)|s(t)\right]. \label{eqn:app_obj} 
\end{eqnarray}
where $\pi=\{a(t)\}$ denotes the entire action sets. By choosing $Q^{\pi}(s,a) = \mathbb{E}_{\pi} \big[r(t)| s(t) = s, a(t) = a\big]$ to be the expected Q-value function of taking action $a$ in state $s$ under a policy $\pi$, the optimal policy to solve \eqref{eqn:app_obj}, $\pi^{\star}$, can be obtained as,
\begin{eqnarray}
\pi^{\star} =  \underset{\pi}{\arg \max} \ Q^{\pi}(s,a).
\end{eqnarray}
According to \cite{32}, we can get the Bellman equation of $Q^{\pi}(s,a)$.
\begin{align}
Q^{\pi}(s,a)=\mathbb{E}_{\pi}\left[\right.r(t+1)+\gamma Q^{\pi}(s(t+1),a(t+1))\nonumber\\
|s(t)=s,a(t)=a\left.\right].
\end{align}
Although standard values and policy iterations can be used to find the optimal solution \cite{21}. However, due to the large state space, it is difficult to realize the computational complexity and memory size of the optimal solution. Therefore, we consider an N-step horizon expressed as follows

\begin{eqnarray}
\underset{a(t^{\prime})}{\textrm{maximize}}  && \sum_{t=t^{\prime}}^{t^{\prime}+N}\mathbb{E}\left[r(t)|s(t^{\prime})\right]. \label{eqn:app_obj} 
\end{eqnarray}
\section{Appendix B }
\label{appd:thm2}
Simplification of Problem~\ref{prob:equ}. 
 
Assuming that $L$ packets are transmitted and it is sufficient to be able to transmit the maximum capacity of all links. According to equation \eqref{eqn:P_OUT} and equation \eqref{eqn:RLC_Q}, $Q_{P/s}^R(t^{\prime})$ are equivalent to the following expressions
 \begin{eqnarray} 
 \vert Q_{P/s}^R(t^{\prime}+1)\vert&=&\vert Q_{P/s}^R(t^{\prime})\vert\nonumber\\&&+A_{P/s}(t)-\vert S_{P/s}^M(t)\vert.
\end{eqnarray}
We define $\vert Q_{P/s}^{\star}(t^{\prime})\vert$ and $\vert Q_{P/s}^{\star}(t^{\prime})\vert$ to represent the size of the PCC and $s^{th}$ SCC RLC buffers in the next $N$ time slots.
 \begin{eqnarray}
&&\vert Q_{P/s}^{\star}(t^{\prime})\vert=\vert Q_{P/s}^{R}(t^{\prime})\vert+\nonumber\\
&&\sum_{t=t^{\prime}}^{t^{\prime}+N}\mathbb{E}\left[\vert Q_{P/s}^R(t+1)\vert-\vert Q_{P/s}^R(t)\vert\right].
\end{eqnarray}
According to the relationship between receiving and sending data packets in the RLC buffer. It can be rewritten as 
\begin{eqnarray}
\vert Q_{P/s}^{\star}(t^{\prime})\vert&=&\vert Q_{P/s}^{R}(t^{\prime})\vert+\nonumber\\
&&\sum_{t=t^{\prime}}^{t^{\prime}+N}\mathbb{E}\left[A_{P/s}(t)-\vert S_{P/s}^M(t)\vert\right]. \label{eqn:Q_e}
\end{eqnarray}
The  RLC buffer status reflects the dynamic relationship between the  transmission strategy $A_{P/s}(t)$ and the transmission capacity $S_{P/s}^M(t)$.
 If the size of packets sent to one RLC buffer is larger than its transmission capacity, it will lead to packet accumulation in the RLC buffer that may cause bufferbloat problem \cite{33}. And the amount of packets in the other buffer is insufficient that leads to the actual transmission capacity is far lower than the transmission capacity.

In order to observe whether the transmission strategy matches the transmission capacity of the different link. We let $\Delta H(t^{\prime})=\vert Q_P^{\star}(t^{\prime})\vert-\sum_{s=1}^{N_{SCC}}\vert Q_s^{\star}(t^{\prime})\vert$. And we consider two cases one is that there are excessive data packets sent to the PCC, and the other is that there are excessive data packets sent to the SCCs.

\begin{case}
\label{case1}
There is an excessive number of data packets sent to the PCC. 
According to equation \eqref{eqn:M_P}, $S_{P}^M(t)=\lfloor \rho_{P}/\rho_{s} \rfloor$. For SCCs, the arriving data packets $A_s(t)$ is less than its transmission capacity $S_{s}^M(t)$ for $N$ time slots, the packets in RLC buffer will not accumulate. Thus, the size of SCC is as follows:
\begin{eqnarray} 
 \vert Q_s^{\star}(t^{\prime})\vert&\approx&0. \label{eqn:Q_s}
\end{eqnarray}
From equation \eqref{eqn:f_func} we can get the data packet received by the PDCP layer for $N$ slots is as follows:
\begin{eqnarray} 
\sum_{t=t^{\prime}}^{t^{\prime}+N}
\mathbb{E}\left[  \vert Q_U^{P}(t)\vert\right]=\sum_{t=t^{\prime}}^{t^{\prime}+N}\mathbb{E}\big[\vert S_{P}^M(t)\vert\nonumber\\+\sum_{s=1}^{N_{SCC}}\vert S_s^M(t)\vert\big].
\end{eqnarray}
From equation \eqref{eqn:Q_e},we can get the folowing equation for PCC:
\begin{eqnarray}
&&\sum_{t=t^{\prime}}^{t^{\prime}+N}\mathbb{E}\left[\vert S_{P}^M(t)\vert\right]=-\vert Q_P^{\star}(t^{\prime})\vert\nonumber\\ 
&&\quad+\vert Q_{P}^{R}(t^{\prime})\vert+\sum_{t=t^{\prime}}^{t^{\prime}+N}\mathbb{E}\left[ A_{P}(t)\right].
\end{eqnarray}
And from equation \eqref{eqn:Q_e} and  \eqref{eqn:Q_s}, we can get the folowing equation for $N_{SCC}$ SCCs:
\begin{eqnarray} 
&&\sum_{t=t^{\prime}}^{t^{\prime}+N}\sum_{s=1}^{N_{SCC}}\vert S_s^M(t)\vert=-\sum_{s=1}^{N_{SCC}}\vert Q_s^{\star}(t^{\prime})\vert\nonumber\\
&&+ \sum_{s=1}^{N_{SCC}}\vert Q_{s}^{R}(t^{\prime})+\sum_{t=t^{\prime}}^{t^{\prime}+N}\mathbb{E}\left[\sum_{s=1}^{N_{SCC}}A_{s}(t)\right]\notag\\
&&=\sum_{s=1}^{N_{SCC}}\vert Q_{s}^{R}(t^{\prime})+\sum_{t=t^{\prime}}^{t^{\prime}+N}\mathbb{E}\left[\sum_{s=1}^{N_{SCC}}A_{s}(t)\right].
\end{eqnarray}
Thus, the throughput of PDCP layer for $N$ time slots can be rewritten as follow:
\begin{eqnarray} 
&&\sum_{t=t^{\prime}}^{t^{\prime}+N}\mathbb{E}\left[  \vert Q_U^{P}(t)\vert\right]=-\vert Q_P^{\star}(t^{\prime})\vert+\vert Q_{P}^{R}(t^{\prime})\vert+\notag\\
&&\sum_{s=1}^{N_{SCC}}\vert Q_{s}^{R}(t^{\prime})\vert +\sum_{t=t^{\prime}}^{t^{\prime}+N}\mathbb{E}\left[A_{P}(t)+\sum_{s=1}^{N_{SCC}}A_{s}(t)\right].\notag\\
\end{eqnarray}
For $N$ time slots, $L$ packets are sent and $\vert Q_P^{\star}(t^{\prime})\vert=\Delta H(t^{\prime})$. Therefore, the expression of $\sum_{t=t^{\prime}}^{t^{\prime}+N}\mathbb{E}\left[  \vert Q_U^{P}(t)\vert\right]$  can finally be written as
\begin{eqnarray} 
\sum_{t=t^{\prime}}^{t^{\prime}+N}\mathbb{E}\left[ \vert Q_U^{P}(t)\vert\right]=-\Delta H(t^{\prime})+L.
\end{eqnarray}
\end{case}
\begin{case} In the case of excessive data packets sent to SCCs. Similarly, $S_{s}^M(t)=1$. For PCC, the arriving data packet $A_P(t)$ is less than its transmission capacity $S_{P}^M(t)$  for $N$ slots. Thus there are the following expressions
\begin{eqnarray} 
Q_P^{\star}(t^{\prime})& \approx&0.\label{eqn:Q_p}
\end{eqnarray}
Similar to the derivation process of Case \ref{case1}, from equation \eqref{eqn:Q_e} and  \eqref{eqn:Q_p} we can get the throughput of the PDCP layer for N time slots:
\begin{eqnarray} 
&&\sum_{t=t^{\prime}}^{t^{\prime}+N} \mathbb{E}\left[  \vert Q_U^{P}(t)\vert\right]=-\sum_{s=1}^{N_{SCC}}\vert Q_s^{\star}(t^{\prime})\vert+\vert Q_{P}^{R}(t^{\prime})\vert+\notag\\&&
\sum_{s=1}^{N_{SCC}}\vert Q_s^{R}(t^{\prime})\vert
+\sum_{t=t^{\prime}}^{t^{\prime}+N}\mathbb{E}\left[ A_{P}(t)+\sum_{s=1}^{N_{SCC}}A_{s}(t)\right]\notag\\
&& =\Delta H(t^{\prime})+L.
\end{eqnarray}
\end{case}
We take the absolute value of $\vert \Delta H(t^{\prime})\vert$. Thus, the sum of  $  \vert Q_U^{P}(t)\vert$ for $N$ time slots is negatively correlated with $\vert\Delta H(t^{\prime})\vert $. Problem \ref{prob:equ} can be equivalent to the following form.
\begin{eqnarray}
\underset{A_{P/s}(t^{\prime})}{\textrm{minimize}}  && \vert\Delta H(t^{\prime})\vert. 
\end{eqnarray}
For the  expression of $\Delta H(t^{\prime})$, there are the following
expressions,
\begin{eqnarray}
&&\Delta H(t^{\prime})=Q_P^{\star}(t^{\prime})-\sum_{s=1}^{N_{SCC}}Q_s^{\star}(t^{\prime})\notag\\
&&=\vert Q_P^{R}(t^{\prime})\vert-\sum_{s=1}^{N_{SCC}}\vert Q_s^{R}(t^{\prime})\vert\notag\\
&&+\sum_{t=t^{\prime}}^{t^{\prime}+N}\mathbb{E}\Bigg[\vert Q_{P}^R(t+1)\vert
-\sum_{s=1}^{N_{SCC}}\vert Q_{s}^R(t+1)\vert\notag\\
&&-(\vert Q_{P}^R(t)\vert-\sum_{s=1}^{N_{SCC}}\vert Q_{s}^R(t)\vert)\Bigg]\notag\\
&&=B(t^{\prime})+\sum_{t=t^{\prime}}^{t^{\prime}+N}\mathbb{E}\left[ B(t+1)-B(t)\right].\label{eqn:Q_B}
\end{eqnarray}
Assuming that the environmental parameters and the transmission strategy  after time slot $t^{\prime}$ remains constant, the formula for equation \eqref{eqn:Q_B} can be expanded as follows
\begin{eqnarray}
&&\Delta H(t^{\prime})=B(t^{\prime})+\sum_{t=t^{\prime}}^{t^{\prime}+N}\mathbb{E}\left[ B(t^{\prime}+1)-B(t^{\prime})\right]\notag\\
&&=B(t^{\prime})+ B(t^{\prime}+1)-B(t^{\prime})+B(t^{\prime}+2)\notag\\
&&-B(t^{\prime}+1)+...+B(t^{\prime}+N+1)-B(t^{\prime}+N)\notag \\
&&=B(t^{\prime})+(N+1)\cdot\bigg[A_{P}(t^{\prime})-\vert S_{P}^M(t)\vert-\notag \\
&&\quad\sum_{s=1}^{N_{SCC}}\left(A_{s}(t^{\prime})-\vert S_{s}^M(t)\vert\right)\bigg]\notag\\
&&=B(t^{\prime})+(N+1)\cdot\bigg[A_{P}(t^{\prime})-\sum_{s=1}^{N_{SCC}}A_{s}(t^{\prime})-\notag\\
&&\quad(\lfloor \rho_{P}/\rho_{s} \rfloor-N_{SCC})\bigg].
\end{eqnarray}
The optimization problem is then transformed into the following form,
\begin{eqnarray}
\underset{A_{P/s}(t^{\prime})}{\textrm{minimize}}  && \big\vert B(t^{\prime})+(N+1)\cdot\big[A_{P}(t^{\prime})-\notag\\&&\sum_{s=1}^{N_{SCC}}A_{s}(t^{\prime})
-(\lfloor \rho_{P}/\rho_{s} \rfloor-N_{SCC})\big]\big\vert\notag.\\ 
\end{eqnarray}

\end{appendices}

\bibliographystyle{IEEEtran}
\bibliography{references}
% \newpage
% \biographies
\end{document}